\documentclass[11pt]{article}
\usepackage[left=1in, top=1in, right=1in, bottom=1in]{geometry}
\usepackage{amsmath, amssymb, amsthm, amsfonts}
\usepackage[linesnumbered,ruled,vlined]{algorithm2e}
\usepackage{tikz}
\usepackage{caption}
\usepackage{bbm}
\usepackage{enumerate}
\usepackage{enumitem}
\usepackage[T1]{fontenc}

\DeclareMathOperator*{\E}{\mathbf{E}}
\let\Pr\relax
\DeclareMathOperator*{\Pr}{\mathbf{Pr}}
\newtheorem{theorem}{Theorem}
\newtheorem{definition}{Definition}
\newtheorem{example}{Example}
\newtheorem{lemma}{Lemma}
\newtheorem{prop}{Proposition}

\title{Predict and Match: Prophet Inequalities with Uncertain Supply}
\date{}

\newcommand*\samethanks[1][\value{footnote}]{\footnotemark[#1]}

\author{Reza Alijani\thanks{Duke University. Email: \texttt{\{alijani,kamesh,knwang\}@cs.duke.edu}.} \and Siddhartha Banerjee\thanks{Cornell University. Email: \texttt{sbanerjee@cornell.edu}.} \and Sreenivas Gollapudi\thanks{Google Research. Email: \texttt{sgollapu@google.com}.} \and Kamesh Munagala\samethanks[1] \and Kangning Wang\samethanks[1]}

\begin{document}

\thispagestyle{empty}
\maketitle

\pagenumbering{gobble}
\begin{abstract}
\noindent We consider the problem of selling perishable items to a stream of buyers in order to maximize social welfare. 
A seller starts with a set of identical items, and each arriving buyer wants any one item, and has a valuation drawn {\em i.i.d.} from a known distribution.
Each item, however, disappears after an {\em a priori} unknown amount of time that we term the {\em horizon} for that item. 
The seller knows the (possibly different) distribution of the horizon for each item, but not its realization till the item actually disappears. 
As with the classic prophet inequalities, the goal is to design an online pricing scheme that competes with the prophet that knows the horizon and extracts full social surplus (or welfare).

Our main results are for the setting where items have independent horizon distributions satisfying the monotone-hazard-rate (MHR) condition. Here, for any number of items, we achieve a constant-competitive bound via a conceptually simple policy that balances the rate at which buyers are accepted with the rate at which items are removed from the system. We implement this policy via a novel technique of matching via probabilistically simulating departures of the items at future times. 
Moreover, for a single item and MHR horizon distribution with mean $\mu$, we show a tight result: There is a fixed pricing scheme that has competitive ratio at most $2 - 1/\mu$, and this is the best achievable in this class. 

We further show that our results are best possible. First, we show that the competitive ratio is unbounded without the MHR assumption even for one item. Further, even when the horizon distributions are {\em i.i.d.} MHR and the number of items becomes large, the competitive ratio of any policy is lower bounded by a constant greater than $1$, which is in sharp contrast to the setting with identical deterministic horizons.
\end{abstract}

\newpage
\pagenumbering{arabic}
\setcounter{page}{1}

\section{Introduction}
\noindent Online posted pricing problems are one of the canonical examples in online decision-making and optimal control. The basic model comprises of a fixed supply of non-replenishable items; buyers (demand) arrive in an online fashion over a fixed time interval, and the platform sets prices to maximize some objective such as social surplus (welfare) or revenue. Another variant of this setting is found in internet advertising, where the number of advertisements (supply) is assumed to be fixed (for example, based on contracts between the publisher and advertisers), while keywords/impressions (demand) arrive online, and are matched to ads via some policy. The demand is typically assumed to obey some underlying random process, which allows the problem to be cast as a Markov Decision Process (MDP); however, in many settings, such a formulation suffers from a ``curse of dimensionality'', making it infeasible to solve optimally.

An important idea for circumventing the computational intractability of optimal pricing is that of \emph{prophet inequalities} --- heuristics with performance guarantees with respect to the optimal policy in hindsight (\emph{i.e.}, the performance of a \emph{prophet} with full information of future arrivals). 
The simplest prophet inequality has its origins in the statistics community~\cite{krengel1977semiamarts} --- given a single item and $T$ arriving buyers with values drawn from known distributions, there is a pricing scheme using only a single price that extracts at least half the social surplus earned by the prophet (moreover, this is tight). 
More recently, 
there has been a long line of work generalizing this setting to incorporate multiple (possibly non-identical)
items, as well as combinatorial buyer valuations~\cite{HajiaghayiKP,chawla2010multi,kleinberg2012matroid,feldman2014combinatorial,DuettingFKL17,
rubinstein2017combinatorial,correa2017posted,abolhassani2017beating,ehsani2018prophet}.

The aim of our work is to develop a theory of prophet inequalities for settings with \emph{uncertainty in future supply}. This is a natural extension of the basic posted-price setting, and indeed special cases of our framework have been considered before~\cite{samuel1996optimal,HajiaghayiKP}~(in the context of optimal secretary problems with a random ``freeze'' on hiring).  What makes these problems of greater relevance today is the rise of online `sharing economy' marketplaces, such as those for transportation (Lyft, Uber), labor (Taskrabbit, Upwork), lodging (Airbnb), medical services (PlushCare), \emph{etc}.  The novelty in such marketplaces arises because of their {\em two-sided nature}: in addition to buyers who arrive online, the supply is now controlled by ``sellers'' who can arrive and depart in an online fashion. For example, in the case of ridesharing/lodging platforms, the units of supply (empty vehicles/vacant listings) arrive over time, and have some {\em patience} interval after which they abandon the system (get matched to rides on other platforms/remove their listings).  
Supply uncertainty also arises in other settings, for instance, if items are {\em perishable} and last for {\em a priori} random amounts of time.
Our work aims to understand the design of pricing policies for such settings, and characterize how the resulting prophet inequalities depend on the characteristics of the supply uncertainty.

\subsection{Model}


\noindent We introduce ``supply uncertainty'' into the basic prophet inequality setting as follows: There are $m$ items present initially, but these do not last till the end of the buyer arrivals, but instead, depart after an {\em a priori} unknown amount of time. Formally, we assume each item $i$ samples a {\em horizon} from a distribution $H_i$, at which time it departs. We assume the horizon lengths for items are mutually independent, and also independent of the valuation distribution of the buyers. Note though that the items can have different horizon distributions.
We denote the maximum possible horizon length for any item as $n$. 

On the demand side, we assume there is an infinite stream of unit-demand buyers arriving online, where the valuation of the $h$-th arriving buyer is a random variable $X_h$ drawn \emph{i.i.d.} from a distribution $V$. From the perspective of a buyer, all items are interchangeable, and hence being matched to any item that has not yet departed yields value $X_h$.
Note that assuming an infinite stream of buyers is without loss of generality, because we can encode any upper bound on the number of buyers in the horizon distributions.

The algorithm designer knows the horizon distribution $H_i$ for each item, and the buyer value distribution $V$, but not the realized horizons for each item (until the item actually departs), or the value for any buyer.
The goal is to design an online pricing scheme that competes with a prophet that knows the realized horizons of each item and the valuation sequence of buyers, and extracts full social surplus (or welfare). 

The main outcome of the standard prophet inequality is that there are constant-competitive algorithms for maximizing welfare, even when buyers are heterogeneous and arrive in arbitrary order. This however turns out to be impossible in the presence of item horizons without additional assumptions. First, even with 
\emph{i.i.d.} horizons, achieving a constant factor turns out to be impossible for general horizon distributions (\emph{cf.} Theorem~\ref{thm:gen_low_any}); thus to make progress, we need more structure on the horizons. One natural assumption is that each item is more and more likely to depart as time goes on, which can be formalized as follows. 
\begin{definition}
\rm
A horizon distribution $H$ satisfies the \emph{monotone-hazard-rate (MHR) condition} if:
\[
\Pr_{h \sim H}[h \geq h^* + 2 \mid h \geq h^* + 1] \leq \Pr_{h \sim H}[h \geq h^* + 1 \mid h \geq h^*], \quad\forall\,h^* \geq 1.
\]
\end{definition}
Several distributions satisfy the MHR condition, including uniform, geometric, deterministic, and Poisson; note also that truncating an MHR distribution preserves the condition.

Finally, even with MHR horizons, buyer heterogeneity is a barrier for obtaining a constant-competitive algorithm, as demonstrated by the following example, with {\em deterministic valuations and known order of arrivals}.
\begin{example}
\label{eg1}
\rm
Given $m=1$ item with horizon following a geometric distribution with parameter $0.5$, consider a sequence of $n$ buyers with $v_h=2^h$ for $h = 1,2,\ldots,n$. The expected value of the prophet is $\Theta(n)$ while any algorithm can only achieve a constant value in expectation.
\end{example}


\subsection{Main Result: Prophet Inequalities under Uncertain Supply}

\noindent The above discussion motivates us to study settings with \emph{i.i.d.} buyers, and items with MHR horizons. 
Our main result is that these two assumptions are \emph{sufficient to obtain a constant-competitive approximation to the prophet welfare}. In particular, our main technical result is the following theorem, which we prove in Section~\ref{sec:multi_MHR}.

\begin{theorem} 
\label{thm:main}
There is a constant-competitive online policy for social surplus for any $m \ge 1$ items with independent and {\em possibly non-identical} MHR horizon distributions, and unit-demand buyers arriving with \emph{i.i.d.} valuations.
\end{theorem}
Though the complete algorithm is somewhat involved, at a high level, it is based on a simple underlying idea: to be constant-competitive against the prophet, we need to choose prices so as to balance the rate of matches and departures. 
Achieving this in the general case is non-trivial, and requires some new technical ideas.
However, for the special case of a single item, balancing can be achieved via a simple fixed pricing scheme. In Section~\ref{sec:mhrSingle}, we use this to obtain the following tight result for the $m = 1$ setting (this also serves as a primitive for our overall algorithm): 
\begin{theorem} 
	\label{thm:single}
	There is a fixed pricing scheme for a single item with an MHR horizon distribution with mean $\mu$ that has competitive ratio $2 - 1/\mu$. Further, this bound is tight for the geometric horizon distribution with mean $\mu$.
\end{theorem}
Intuitively, the factor of two in the above theorem corresponds to the prophet considering matching and departures as the same, which an algorithm cannot do. The surprising aspect is that this simple policy is worst-case optimal within the class of instances with MHR horizons --- this is in contrast to deterministic horizons, where fixed pricing is known to be suboptimal for the special case of one item with known (deterministic) horizon and {\em i.i.d.} buyers~\cite{hill1982comparisons, ehsani2018prophet}.

\subsection{Lower Bounds} 
\noindent We complement our positive results by showing several lower bounds that establish their tightness. 
As mentioned above, in Section~\ref{sec:mhrSingle}, we show a (tight) lower bound of $2 - 1/\mu$ for $m = 1$ items with MHR horizons.  
Our main lower bounds in Section~\ref{sec:multi} generalizes this to $m \ge 1$ items. 

\begin{theorem}
	\label{thm:lb}
	For the multi-item setting with \emph{i.i.d.} geometric horizons: 
	\begin{itemize}
		\item For any number of items, there is a lower bound of $1.57$ on the competitive ratio of any dynamic pricing scheme; in the limit when the number of items goes to infinity, this improves to $2$.
		\item No fixed pricing scheme can be $o(\log \log m)$-competitive where $m$ is the number of items.
	\end{itemize}
\end{theorem}
The above theorem implies that the MHR horizon setting, even with \emph{i.i.d.} horizons, is significantly different from the setting with multiple items and a single deterministic horizon (where fixed pricing extracts $\left(1 - O\left(\frac{1}{\sqrt{m}}\right)\right)$-fraction of surplus~\cite{alaei2014bayesian}). Put differently, the lower bound emphasizes that even with \emph{i.i.d.} horizons, to obtain a constant-competitive algorithm, it is {\em not} sufficient to replace the horizon distributions by their expectations and use standard prophet inequalities --- the stochastic nature of the horizons allows for significant deviations in the order of departures of the items, and a policy that knows this ordering can potentially extract much more welfare. Given this, it is quite surprising that a simple dynamic pricing scheme achieves a constant approximation.


Finally, we consider the general case where there is no restriction on the horizon distribution. In this setting, the presence of supply uncertainty severely limits the performance of any non-anticipatory dynamic pricing scheme in comparison to the omniscient prophet. In particular, we show that for any number of items and {\em i.i.d.} buyer valuations, the ratio between the welfare of any algorithm and the prophet grows with the horizon, \emph{even if the algorithm knows the realized valuations}.
\begin{theorem}
	\label{thm:gen_low_any}
	For any $m \ge 1$ items, there exists a family of instances such that the prophet has welfare $\Omega\left(\frac{\log n}{\log \log n}\right)$-factor larger than any  online policy, even if the policy knows all the realized values, but not the realized horizons. Here,  $n = \max_i \{\mathrm{supp}(H_i)\}$. 
\end{theorem}
This generalizes similar lower bounds for settings where the horizon is unknown~\cite{hill1991minimax,HajiaghayiKP}. The proof of this result is provided in Appendix~\ref{app:gen_low_any}.

\subsection{Technical Highlights}

\noindent At a high level, we achieve our results via a conceptually simple and natural class of {\em balancing policies} that generalizes policies for the deterministic-horizon case:

\begin{quote}
{\bf Balancing Policy.} Balance the rate at which buyers are accepted to the rate at which items depart the system because their horizon is reached.
\end{quote}

Converting this high-level description of balancing into a concrete policy requires new technical ideas.
We first note the technical challenges we encounter. In the setting with deterministic identical horizons~\cite{krengel1977semiamarts,feldman2014combinatorial}, we can achieve constant-competitive algorithms (or even better) via a {\em global} expected value relaxation that yields a fixed pricing scheme. Indeed, such an argument can safely assume buyers are non-identical with adversarial arrival order. However, the setting with stochastic horizons is very different. First, as Example~\ref{eg1} shows, even for $m =1$ item with geometric horizon, there is an $\Omega(n)$ lower bound when buyer valuations are not identically distributed. Secondly, for $m > 1$ items, we need dynamic pricing even in the simplest settings --- when horizons are \emph{i.i.d.} geometric (see Theorem~\ref{thm:lb}), or when they are deterministic. This precludes the use of a global one-shot analysis. 

At this point, we could try using techniques from stochastic optimization, particularly stochastic matchings~\cite{Immorlica,Bansal} and multi-armed bandits~\cite{GuhaM07,GuptaKMR11}. Here, the idea is to come up with a {\em weakly coupled} relaxation, say one policy per item, and devise a feasible policy by combining these.  
However, these algorithms crucially require the state of the system to only change via policy actions, and our problem more is similar to a {\em restless bandit} problem~\cite{GuhaMS10} where item departures cause the state of the system can change {\em regardless} of policy actions taken. Indeed, the actual departure process itself may significantly deviate from its expected values, making it non-trivial to use a global relaxation. 

\paragraph{Simulating Departures.} This brings up our technical highlight: {\em Instead of encoding the departure process in a fine-grained way into a relaxation, we simulate its behavior in our final policy.}  In more detail, we first write a weak relaxation of the prophet's welfare separately in a sequence of stages  with geometrically decreasing number of items. This only uses the expected number of items that survive in the stage, and not the identity of these items. The advantage of such a weak relaxation is that it yields a solution with nice structure: this policy non-adaptively sets a fixed price in each stage to balance the departure rate with the rate of matches. However, it is non-trivial to construct a feasible policy from this relaxation, since the relaxation decouples the allocations of the prophet across different stages, while any feasible algorithm's allocations are clearly coupled. Indeed, the optimal feasible policy is the solution to a dynamic program with state space exponential in $m$, and the prophet is further advantaged by knowing which items depart earlier in the future. 

Surprisingly, we show that our simple relaxation is still enough to achieve a constant-competitive algorithm. We do so by simulating the departure process, that is, by choosing items for matching with the same probability that they would have departed at a future point in time. This couples the stochastic process that dictates the number of items available in the policy with that in the prophet's upper bound, albeit with a constant-factor speedup in time. This yields a {\em non-adaptive} policy that makes its pricing decisions for the entire horizon, as well as the (randomized) sequence in which to sell the items, in advance. 
We believe such a policy construction that simulates the evolution of state of the system may find further applications in the analysis of  restless MDPs.

\paragraph{Lower Bounds from Time-Reversal.}
Our lower bounds are all based on demonstrating particular bad settings as in Example~\ref{eg1}. From a technical perspective, the most interesting construction is that in Theorem~\ref{thm:lb} --- here, we first consider a canonical, asymptotic regime where the horizon distribution is geometric with mean approaching infinity, and show that we can closely approximate the behavior of the prophet and the algorithm via an appropriate Markov chain. We then define and analyze a novel time-reversed Markov chain encoding the prophet's behavior, that captures matching a departing item to the optimal buyer that arrived previously.  

 \subsection{Related Work}
\label{sec:related}
\noindent The first prophet inequalities are due to Krengel and Sucheston~\cite{krengel1977semiamarts,krengel1978semiamarts}. It was subsequently shown~\cite{samuelcahn1}, there is a $2$-competitive fixed pricing scheme that is oblivious to the order in which the buyers arrive, and this ratio is tight in the worst case over the arrival order. Motivated by applications to online auctions, since then there have been several extensions to multiple items~\cite{kennedy1987prophet,HajiaghayiKP,alaei2014bayesian}, matching setting~\cite{alaei2012online, truong2019prophet}, matroid constraints~\cite{kleinberg2012matroid} and general combinatorial valuation functions~\cite{feldman2014combinatorial,rubinstein2017combinatorial}. 

Our work is a generalization of the single-item setting where buyer valuations are {\em i.i.d.} and the horizon is known, to the case where the horizon is stochastic and there are multiple items. The setting with known horizons was first considered in Hill and Kertz~\cite{hill1982comparisons}. In this case, the optimal pricing scheme can be computed by a dynamic program, and a sequence of results~\cite{kertz1986stop,abolhassani2017beating,correa2017posted} show a tight competitive ratio of $1.342$ for this dynamic program against the prophet. In contrast, we show that when the horizon is MHR, a simple fixed pricing scheme has  optimal competitive ratio of $2$.

A generalization of the {\em i.i.d.} setting is the recently-introduced {\em prophet secretary} problem where the buyers are not identical, but the order of arrival is a random permutation. In this case, fixed pricing is a tight $\frac{e}{e - 1}$-approximation~\cite{esfandiari2017prophet,ehsani2018prophet}; and a dynamic pricing scheme can beat this bound~\cite{azar2018prophet,correa2019prophet} by a slight amount. Though our results extend to this setting,  it is not the focus of our paper since the {\em i.i.d.}-valuations case is sufficient to bring out our conceptual message.

The random horizon setting has been extensively studied in the context of the classic secretary problem. When the horizon is unknown (that is, no distributional information at all), no constant-competitive algorithm is possible~\cite{hill1991minimax}. In the context of prophet inequalities, the unknown-horizon setting was considered by Hajiaghayi {\em et al.}~\cite{HajiaghayiKP}, who show again that no constant-competitive algorithm is possible. We use a similar example to extend this lower bound to the case where the horizon is stochastic from a known distribution. 



\section{Prophet Inequality for Heterogeneous Items with MHR Horizons}
\label{sec:multi_MHR}

\noindent In this section, we present the proof of Theorem~\ref{thm:main}. We first give an overview of our algorithm.  At a high level, this scheme attempts to balance the rate that items are assigned to buyers and the rate that items naturally depart. In Section~\ref{sec:stages}, we first introduce a way to divide the entire time horizon into disjoint stages in a way such that during the $k$-th stage,  $\frac{m}{2^k}$ items depart in expectation. We then bound the prophet's welfare separately for each stage (Section~\ref{sec:relax}) --- we do so via a relaxation that ignores the identity of the items, and only captures the constraint that the expected number of matches in a stage is at most the expected number of items present at the beginning of that stage. 

The key technical hurdle at this point is that when we make a matching, we do so without knowing exactly when items depart in the future. This changes the distribution of the items available in subsequent stages. 
To get around this, in each stage, we first {\em simulate} the future departure of items, and use this to select items available for matching in the current stage.
In more detail, in Section~\ref{sec:odd}, we split the stages alternately into even and odd stages, and develop an algorithm whose welfare approximates the welfare of the relaxed prophet from the odd stages (and by symmetry, another algorithm that approximates the welfare from the even stages). 

For approximating the welfare from the odd stages, the algorithm re-divides time into a new set of stages corresponding to the odd stages under the old division (See Figure~\ref{fig:faster}).
We then use each new stage to approximate the welfare generated in the corresponding odd stage in the old division; to do so, we sample candidate items for matching in the current stage with the probability they would leave in the subsequent even stage under the old division. Consequently, for every item, the probability of departure during an even stage under the old division is the same as of being selected for matching in the current stage. We show that this process couples the behavior of the algorithm and the benchmark, assuming the departure processes are MHR. Using concentration bounds, we show that this approach yields a constant approximation.

In addition to the above process, our algorithm needs to separately handle any stage of length $1$ (\emph{i.e.}, any single time period where the expected number of available items reduces by at least half), as well as a final stage where the expected number of available items is constant. We show that the welfare in the length $1$ phases is approximated by a blind matching algorithm which matches all incoming buyers (Section~\ref{sec:short}), while the welfare of the final period is approximated by an algorithm that randomly selects only one item for matching at the beginning, and discards the rest (Section~\ref{sec:final}). For the latter setting (\emph{i.e.}, for a single item setting), we present a tight $2$-competitive fixed pricing scheme for the $m = 1$ setting in Section~\ref{sec:mhrSingle}. Finally, the overall algorithm is based on randomly choosing one of the four candidate algorithms (\emph{i.e.}, for approximating the prophet welfare in odd stages, even stages, short stages, and the final stage), with an appropriately chosen distribution.


\subsection{Splitting Time into Stages}
\label{sec:stages}
\noindent As a first step, we divide the time horizon into $s + 1$ stages. The $k$-th stage corresponds to an interval $[\ell_k, r_k)$. 
For $k = 1, 2, \ldots, s$, we define $r_k$ by
\[
r_k \! \! := \min \! \left\{\! t \! + \! 1\!\!  \ :\ \!\! \E[\text{number of remaining items after time }t] \leq \frac{m}{2^k} \! \right\} \!.
\]
Also $\ell_{k+1} := r_k$ for $k = 1, 2, \ldots, s$; $\ell_1 = 1$ and $r_{s+1} = \infty$.

We set $s$ to be the smallest non-negative integer so that $\frac{m}{2^s} \leq 10$, \emph{i.e.}, $s := \max\left(0, \left\lceil \log_2 \frac{m}{10} \right\rceil\right)$. 
Within the first $s$ stages, we separate stages of length $r_k - \ell_k = 1$ from the rest. We term the stages of length at least $2$ as {\sc Long} stages, and those of length $1$ as {\sc Short} stages. We term the stage $s+1$ as the {\sc Final} stage. Note that based on our choice of $s$, the expected number of items which remain in the final stage is at most $10$, and unless $s = 0$, at least $5$ items in expectation survive at one time step earlier into the final stage.

\subsection{Upper Bound on Prophet's Welfare}
\label{sec:relax}
\noindent In this section, we develop a tractable upper bound for the prophet. Let $\textsc{Pro}$ denote the optimal welfare obtainable by the prophet. We term the total welfare of {\sc Pro} in the {\sc Long} stages as {\sc ProLong}, the total welfare in the {\sc Short} stages as {\sc ProShort}, and the welfare in the {\sc Final} stage as {\sc ProFinal}.
Clearly, we have:
\begin{lemma}
$\textsc{Pro} = \textsc{ProLong} + \textsc{ProShort} + \textsc{ProFinal}$.
\label{lem:Pro_Upper}
\end{lemma}

We bound $\textsc{ProLong}$ and {\sc ProShort} separately for each stage. Let $\textsc{Pro}_k$ denote the welfare from stage $k$, so that $\textsc{ProLong} + \textsc{ProShort} = \sum_{k=1}^s \textsc{Pro}_k$. 
\begin{lemma}
For $1 \le k \le s$, we have:
$$\textsc{Pro}_k  \leq  \min\left(r_k - \ell_k, \frac{m}{2^{k-1}}\right) \cdot \E_{v \sim V}[v \mid v \geq p_k],$$ where $p_k$ satisfies $\Pr_{v \sim V}[v \geq p_k]= \min\left(1, \frac{m / 2^{k-1}}{r_k - \ell_k}\right)$.\footnote{The existence of such $p$ is without loss of generality: Let $t = \min\left(1, \frac{m / 2^{k-1}}{r_k - \ell_k}\right)$. When there exists some $p^*$ such that $\Pr[v \geq p^*] > t$ and $\Pr[v > p^*] < t$, we could accept all values greater than $p^*$ and accept $p^*$ with probability $\frac{t - \Pr[v > p^*]}{\Pr[v = p^*]}$.}
\label{lem:ProEarly}
\end{lemma}
\begin{proof}
Fix a stage $k$. Let $W_i$ be the expected welfare that the prophet gets from buyer $i$, and let $y_i$ be the probability that buyer $i$ is matched by the prophet ($\ell_k \leq i < r_k$).

Notice that in expectation, at most $\frac{m}{2^{k-1}}$ items have horizons of at least $\ell_k$ by the definition of stages. Therefore, $
\sum_{i = \ell_k}^{r_k} y_i \leq \frac{m}{2^{k-1}}$.

Let $F_V$ be the CDF of the distribution $V$.  We have $
W_i \leq y_i \cdot \E_{v \sim V}\left[v \ \middle| \ v \geq F_V^{-1}(1 - y_i)\right]$, 
since when buyer $i$ is matched with probability $y_i$, the prophet cannot do better than getting the top $y_i$-percentile of the distribution $V$ from the buyer. With these constraints, we write a relaxation for the welfare of the prophet during stage $k$:
\begin{equation*}
\begin{array}{ll@{}ll}
\text{max}  & &\displaystyle\sum\limits_{i = \ell_k}^{r_k - 1} W_i\\[3ex] 
\text{s.t.}& &\displaystyle W_i \leq y_i \cdot \E_{v \sim V}\left[v \ \middle| \ v \geq F_V^{-1}(1 - y_i)\right],  &\forall i \! = \! \ell_k, \ell_k \! + \! 1, \ldots, r_k \! - \! 1,\\[3ex]&
&\displaystyle\sum\limits_{i = \ell_k}^{r_k - 1} y_i \leq \frac{m}{2^{k-1}},  &\\[3ex]&
&y_i \in [0, 1], &\forall i  \! = \! \ell_k, \ell_k \! + \! 1, \ldots, r_k \! - \! 1.
\end{array}
\end{equation*}
Clearly $y_i$'s should be equal in the optimal solution. Therefore,
\[
\sum\limits_{i = \ell_k}^{r_k - 1} W_i \leq (r_k - \ell_k) \cdot \min\left(1, \frac{m / 2^{k-1}}{r_k - \ell_k}\right) \cdot \E_{v \sim V}[v \mid v \geq p_k],
\]
where  $\Pr_{v \sim V}[v \geq p_k]= \min\left(1, \frac{m / 2^{k-1}}{r_k - \ell_k}\right)$. Summing over the $s$ stages finishes the proof.
\end{proof}
Notice that in our upper bound for $\sum_{k=1}^s \textsc{Pro}_k$, if an item departs during stage $k$, we allow it to be matched once in stage $1$, once in stage $2$, \ldots, and once in stage $k$. However, since the expected number of departures in each stage exponentially decreases, only a constant factor is lost comparing with the finer relaxation where we enforce the constraint that each item is only matched once across the stages. Our coarser relaxation enables a cleaner benchmark to work on.

We next bound $\textsc{ProFinal}$. Let $\textsc{ProSingle}_i$ be the optimal welfare of the prophet (from all stages) if item $i$ is the only item available in the system, {\em i.e.}, the single-item setting. We consider this setting in detail in Section~\ref{sec:mhrSingle}.
\begin{lemma}
$\textsc{ProFinal} \leq \sum_{i = 1}^m \Pr_{h_i \sim H_i}[h_i \geq \ell_{s+1}] \cdot \textsc{ProSingle}_i$.
\label{lem:ProFinal}
\end{lemma}
\begin{proof}
Let $W_i$ be the welfare that the prophet can get from item $i$ during the final stage. We have
\begin{align*}
W_i \leq &\Pr_{h_i \sim H_i}[h_i \text{ reaches the final stage}] \cdot\\
&\quad \quad \E_{h_i \sim H_i}[\text{welfare from item }i\text{ in the final stage} \mid h_i \text{ reaches the final stage}]\\
\leq &\Pr_{h_i \sim H_i}[h_i \text{ reaches the final stage}] \cdot \E[\text{welfare from item }i] \\
= &\Pr_{h_i \sim H_i}[h_i \geq \ell_{s+1}] \cdot \textsc{ProSingle}_i,
\end{align*}
where the second inequality comes from the MHR condition of $H_i$: $\Pr_{h_i \sim H_i}[h_i \geq \ell_{s+1} + k \mid h_i \geq \ell_{s+1} - 1 + k] \leq \Pr_{h_i \sim H_i}[h_i \geq 1 + k \mid h_i \geq k]$ --- item $i$ would depart faster if it started at time $\ell_{s+1}$.

Summing up the items, we have:
\[
\textsc{ProFinal} \leq \sum_{i = 1}^m W_i \leq \sum_{i = 1}^m \Pr_{h_i \sim H_i}[h_i \geq \ell_{s+1}] \cdot \textsc{ProSingle}_i. \qedhere
\]
\end{proof}

Lemmas~\ref{lem:Pro_Upper}, \ref{lem:ProEarly} and~\ref{lem:ProFinal} together give an upper bound for our benchmark as:
\begin{align*}
\textsc{Pro}\leq 
&\left[\sum_{k\leq s,\, r_k - \ell_k > 1}
\textsc{Pro}_k\right] + 
\left[\sum_{k\leq s,\, r_k - \ell_k = 1} \textsc{Pro}_k\right] + \left[\sum_{i = 1}^m \Pr_{h_i \sim H_i}[h_i \geq \ell_{s+1}] \cdot \textsc{ProSingle}_i\right]
\end{align*}
where the three term correspond to an upper bound on the prophet's welfare in the {\sc Long}, {\sc Short} and {\sc Final} stages respectively (\emph{i.e.}, {\sc ProLong}, {\sc ProShort}, and {\sc ProFinal}). In the next three sections, we describe three separate algorithms, each one of which, if run independently, provides an approximation to one of the terms. 
Our overall algorithm is then based on randomly choosing between the three algorithms with appropriately chosen distribution.

\subsection{Approximating {\sc ProLong}: The {\sc DepartureSimulation} Algorithm}
\label{sec:odd}

\noindent We first approximate upper bound given in Lemma~\ref{lem:ProEarly}. Within this, we approximate {\sc ProLong} and {\sc ProShort} separately. We first focus on {\sc ProLong}, since this is technically the most interesting, and postpone approximating {\sc ProShort} to Section~\ref{sec:short}.

\begin{figure}[htbp]
\centering
\begin{tikzpicture}[scale=0.4,yscale=0.6]
\draw [red,ultra thick] (0, 5) to (3, 5);
\draw [blue,ultra thick,dotted] (3, 5) to (4, 5);
\draw [red,ultra thick] (4, 5) to (6, 5);
\draw [blue,ultra thick,dotted] (6, 5) to (11, 5);
\draw [red,ultra thick] (11, 5) to (14, 5);
\draw [blue,ultra thick,dotted] (14, 5) to (17, 5);

\draw [thick] (0, 4.6) to (0, 5.4);
\draw [thick] (3, 4.6) to (3, 5.4);
\draw [thick] (4, 4.6) to (4, 5.4);
\draw [thick] (6, 4.6) to (6, 5.4);
\draw [thick] (11, 4.6) to (11, 5.4);
\draw [thick] (14, 4.6) to (14, 5.4);
\draw [thick] (17, 4.6) to (17, 5.4);

\draw [red,ultra thick] (0, 1) to (3-0.4, 1);
\draw [red,ultra thick] (3-0.4, 1) to (5-0.4-0.4, 1);
\draw [red,ultra thick] (5-0.4-0.4, 1) to (8-0.4-0.4-0.4, 1);

\draw [thick] (0, 0.6) to (0, 1.4);
\draw [thick] (3-0.4, 0.6) to (3-0.4, 1.4);
\draw [thick] (5-0.4-0.4, 0.6) to (5-0.4-0.4, 1.4);
\draw [thick] (8-0.4-0.4-0.4, 0.6) to (8-0.4-0.4-0.4, 1.4);

\node [draw] at (-1.5, 5) {OLD:};
\node at (17.5, 5) {$\cdots$};

\node [draw] at (-1.5, 1) {NEW:};
\node at (8.5-1.2, 1) {$\cdots$};

\node at (1.5, 5.8) {$S_1$};
\node at (3.5, 5.8) {$S_2$};
\node at (5, 5.8) {$S_3$};
\node at (8.5, 5.8) {$S_4$};
\node at (12.5, 5.8) {$S_5$};
\node at (15.5, 5.8) {$S_6$};

\node at (1.5-0.2, 0.2) {$S_1'$};
\node at (4-0.6, 0.2) {$S_3'$};
\node at (6.5-1.0, 0.2) {$S_5'$};

\draw [dashed] (0, 1) to (0, 5);
\draw [dashed] (2.6, 1) to (3, 5);
\draw [dashed] (2.6, 1) to (4, 5);
\draw [dashed] (4.2, 1) to (6, 5);
\draw [dashed] (4.2, 1) to (11, 5);
\draw [dashed] (6.8, 1) to (14, 5);
\end{tikzpicture}
\caption{Redivision of the Time Horizon}
\label{fig:faster}
\end{figure}
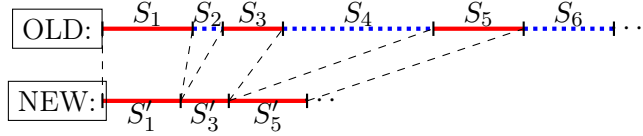

We approximate {\sc ProLong} by Algorithm~\ref{alg:EarlyStages}.  We divide all the $s$ stages into alternate odd and even stages. We focus on illustrating the approximation for odd stages, and that for even stages is identical.  We then re-divide time into stages corresponding to the original odd stages, as illustrated in Figure~\ref{fig:faster}, where $S_k$ stands for the old stage $k$ and $S_k'$ stands for the new stage $k$. At each odd stage, we sample items according to their departure rates during the next (fictitious) even stage. During the new process when items become unavailable by being sampled, each item is as least as likely to survive a stage as before, since the sampling is only as frequent as the natural departures during the original even stages.

Note that we set each $S_k'$ to be $1$ time step shorter than the corresponding $S_k$ and make each fictitious even stage $1$ time step longer (unless the length of $S_k$ is $0$). We do this to ensure enough items will be sampled: Because of integrality constraints, an even stage may be too short (\emph{e.g.}, of length $0$) and if so, little (or nothing if the stage has length $0$) can be sampled there. This is also the reason why \textsc{Short} stages are separately considered.

\begin{algorithm}[htbp]
\DontPrintSemicolon
  $A \gets \{1, 2, \ldots, m\}$  \tcp*{$A = $ Set of available items}
  \For{ each odd stage $k = 1, 3, \ldots,$ till stage $s$}{
      $C_k \gets \varnothing$\tcp*{$C_k = $  Set of items considered in this stage}
      For each $i \in A$,  with probability $\Pr_{h_i \sim H_i}[h_i < r_{k + 1} \mid h_i \geq \ell_{k + 1} - 1]$, place in $C_k$\;
      $A \gets A \setminus C_k$\;
      \If{$r_k - \ell_k \geq 2$}{
        $p_k \gets F_V^{-1} \left(\max\left(0, 1 - \frac{m / 2^{k-1}}{r_k - \ell_k}\right)\right)$\; 
        For each of the next $r_k - \ell_k - 1$ arriving buyers, if this buyer has valuation $\geq p_k$, match to any item in $C_k$ and remove this item from $C_k$\;
          If any item departs, remove it from $A$ and $C_k$\;
        }
      }
\caption{\textsc{DepartureSimulation}: Odd Stages Version}\label{alg:EarlyStages}
\end{algorithm}

Note that Algorithm~\ref{alg:EarlyStages} can be easily modified to work with even stages instead of odd stages, and will yield the corresponding version of the theorem below with ``odd'' replaced by ``even''. In order to show Theorem~\ref{thm:main}, we will use either the odd stages or even stages algorithm depending on which yields larger expected welfare. Note that it is entirely possible that one of these stages yields very low welfare compared to the other.

\begin{theorem}
Algorithm~\ref{alg:EarlyStages} is a $15.1$-approximation to the sum of $\textsc{Pro}_k$ over odd stages $k \le s$ with $r_k - \ell_k \ge 2$. 

\label{thm:early_apx}
\end{theorem}
\begin{proof}
We use $y^+$ to denote $\max(0, y)$. For any odd $k$ with $r_k - \ell_k \geq 2$, let the random variable $M_k$ be the number of items in the set $C_k$ that has horizon of at least $\sum_{k' = 1}^{(k+1)/2} (r_{2k' - 1} - \ell_{2k' - 1} - 1)^+$, \emph{i.e.}, the end of (new) stage $k$. We denote $\sum_{k' = 1}^{j} (r_{2k' - 1} - \ell_{2k' - 1} - 1)^+$ by $\mathcal S_j$  in the rest of the proof. 

$M_k$ is the sum of $m$ independent Bernoulli random variables, where the $i$-th one denotes whether item $i$ is in $C_k$ and has horizon of at least $\mathcal S_{(k+1)/2}$. We have
\begin{align*}
\E[M_k] &= \sum_{i = 1}^m \Pr\left[\text{item }i\text{ is in }C_k\text{ and has horizon of at least } \mathcal S_{(k+1)/2} \right]\\
&= \sum_{i = 1}^m \Pr_{h_i \sim H_i}\left[h_i \geq \mathcal S_{(k+1)/2} \right] \cdot  \left(\prod_{j = 1}^{(k-1)/2}\left(1 - \Pr_{h_i \sim H_i}[h_i < r_{2j} \mid h_i \geq \ell_{2j} - 1]\right)\right) \cdot  \\
& \quad \quad  \Pr_{h_i \sim H_i}[h_i < r_{k + 1} \mid h_i \geq \ell_{k + 1} - 1],
\end{align*}
where we calculate the probability that item $i$ has horizon of at least $\sum_{k' = 1}^{(k+1)/2} (r_{2k' - 1} - \ell_{2k'-1} - 1)^+$, was never selected into $C_{2j - 1}$'s during previous stages $2j - 1 < k$, and was selected into $C_k$. Further simplifying it, we have
\begin{align*}
\E[M_k] &= \sum_{i = 1}^m \left(\prod_{j = 1}^{(k+1)/2}\Pr_{h_i \sim H_i}\left[h_i \geq \mathcal S_{j} \ \middle| \ h_i \geq \mathcal S_{j-1} \right]\right) \cdot  \left(\prod_{j = 1}^{(k-1)/2} \Pr_{h_i \sim H_i}[h_i \geq r_{2j} \mid h_i \geq \ell_{2j} - 1]\right) \cdot \\
& \quad \quad  \Pr_{h_i \sim H_i}[h_i < r_{k + 1} \mid h_i \geq \ell_{k + 1} - 1]
\end{align*}
Since the MHR condition implies the item is more likely to survive in earlier time steps, we have:
\begin{align*}
\E[M_k]  &\geq \sum_{i = 1}^m \left(\prod_{j = 1}^{(k+1)/2}\Pr_{h_i \sim H_i}\left[h_i \geq r_{2j - 1} - 1 \ \middle| \ h_i \geq \ell_{2j - 1}\right]\right) \cdot \left(\prod_{j = 1}^{(k-1)/2} \Pr_{h_i \sim H_i}[h_i \geq r_{2j} \mid h_i \geq \ell_{2j} - 1]\right) \cdot \\
& \quad \quad   \Pr_{h_i \sim H_i}[h_i < r_{k + 1} \mid h_i \geq \ell_{k + 1} - 1]\\
&= \sum_{i = 1}^m \left(\prod_{j = 1}^{(k+1)/2}\Pr_{h_i \sim H_i}\left[h_i \geq \ell_{2j}  - 1 \ \middle| \ h_i \geq \ell_{2j - 1}\right]\right) \cdot \left(\prod_{j = 1}^{(k-1)/2} \Pr_{h_i \sim H_i}[h_i \geq l_{2j + 1} \mid h_i \geq \ell_{2j} - 1]\right) \cdot \\
& \quad \quad \Pr_{h_i \sim H_i}[h_i < r_{k + 1} \mid h_i \geq \ell_{k + 1} - 1]\\
&= \sum_{i = 1}^m \Pr_{h_i \sim H_i}[\ell_{k + 1} - 1 \leq h_i < r_{k + 1}]
\end{align*}
Now, $\sum_{i = 1}^m \Pr_{h_i \sim H_i}[h_i \geq \ell_{k + 1} - 1] \geq \frac{m}{2^{k}}$ and $\sum_{i = 1}^m \Pr_{h_i \sim H_i}[h_i \geq r_{k + 1}] \leq \frac{m}{2^{k + 1}}$. Thus,
\[
\E[M_k] \geq \left(\sum_{i = 1}^m \Pr_{h_i \sim H_i}[h_i \geq \ell_{k + 1} - 1]\right) - \left(\sum_{i = 1}^m \Pr_{h_i \sim H_i}[h_i \geq r_{k + 1}]\right) \geq \frac{m}{2^{k + 1}}.
\]
Note that $\frac{m}{2^{k + 1}} \geq \frac{m}{2^s}$ for $k < s$. For $k = s$, $\E[M_k] \geq \sum_{i = 1}^m \Pr_{h_i \sim H_i}[h_i \geq \ell_{k + 1} - 1] \geq \frac{m}{2^k} = \frac{m}{2^s}$. Thus, $\E[M_k] \geq \frac{m}{2^s} > 5$ for any $k \leq s$. By Chernoff bound,
\[
\Pr\left[M_k \geq \frac{1}{4} \cdot \frac{m}{2^k}\right] \geq 1 - \left(\frac{e^{-0.5}}{0.5^{0.5}}\right)^5 > 0.535.
\]

Now let $p_k$ be $F_V^{-1} \left(\max\left(0, 1 - \frac{m / 2^{k-1}}{r_k - \ell_k}\right)\right)$ where $F_V$ is the CDF of distribution $V$, just as in Algorithm~\ref{alg:EarlyStages}. Let the random variable $N_k$ denote the number of buyers with valuation of at least $p_k$ among the next $r_k - \ell_k - 1$ buyers. We have
\[
\E[N_k] = (r_k - \ell_k - 1) \cdot \min\left(1, \frac{m/2^{k-1}}{r_k - \ell_k}\right).
\]

If $\frac{m/2^{k-1}}{r_k - \ell_k} \geq 1$, then $p_k = -\infty$ and $N_k = r_k - \ell_k - 1$ with probability $1$. In this case, Algorithm~\ref{alg:EarlyStages} gets at least $\min(M_k, r_k - \ell_k - 1) \E[V] \geq \min\left(M_k, \frac{1}{2}(r_k - \ell_k)\right) \E[V]$ in this stage. Since $M_k \geq \frac{1}{4} \cdot \frac{m}{2^k} \geq \frac{1}{8} \cdot (r_k - \ell_k)$ with probability at least $0.535$, we know Algorithm~\ref{alg:EarlyStages} gets at least $\frac{0.535}{4} \cdot (r_k - \ell_k) \cdot \E[V]$ and thus is an $\frac{4}{0.535} < 8$-approximation during the stage.

If $\frac{m/2^{k-1}}{r_k - \ell_k} < 1$, then $r_k - \ell_k > 10$ and $\E[N_k] = (r_k - \ell_k - 1) \cdot \frac{m/2^{k-1}}{r_k - \ell_k} > 0.9 \cdot \frac{m}{2^{k - 1}} = 1.8 \cdot \frac{m}{2^k} > 9$. By Chernoff bound,
\[
\Pr\left[N_k \geq \frac{1}{4} \cdot \frac{m}{2^k}\right] \geq 1 - \left(\frac{e^{-\left(1 - \frac{1}{4 \times 1.8}\right)}}{\left(\frac{1}{4 \times 1.8}\right)^{\frac{1}{4 \times 1.8}}}\right)^9 > 0.994.
\]
When $\min(N_k, M_k) \geq \frac{1}{4} \cdot \frac{m}{2^k}$, Algorithm~\ref{alg:EarlyStages} gets at least $\frac{1}{8}$ the benchmark during the stage. Therefore, it is an $\frac{8}{0.535 \cdot 0.994} < 15.1$-approximation.
\end{proof}

\subsection{Approximating {\sc ProShort}}
\label{sec:short}
\noindent In this section, we deal with length-$1$ stages using Algorithm {\sc BlindMatch}, that simply matches each arriving buyer $i$ to any available item.

\begin{theorem}
Algorithm {\sc BlindMatch} is a $2.3$-approximation to $\sum_{k = 1}^s \textsc {Pro}_k \cdot \mathbbm{1}_{(r_k - \ell_k = 1)} = \E[V] \cdot |\{k \in \{1, 2, \ldots, s\} \mid r_k - \ell_k = 1\}|$.
\label{thm:short_apx}
\end{theorem}
\begin{proof}
Let $z = |\{k \in \{1, 2, \ldots, s\} \mid r_k - \ell_k = 1\}|$, the number of length-$1$ stages. Consider the time $t = \left\lceil \frac{z}{2} \right\rceil$. Since there are still at least $\left\lfloor \frac{z}{2} \right\rfloor$ length-$1$ stages after time $t$, at least $5 \cdot 2^{\left\lfloor \frac{z}{2} \right\rfloor} \geq 5 \cdot \left\lceil \frac{z}{2} \right\rceil$ items in expectation have horizons of at least $t$, by the definition of the stages. Using Chernoff bound, the probability that at least $\left\lceil \frac{z}{2} \right\rceil$ items with horizons of at least $t$ is greater than $1 - \left(\frac{e^{-0.8}}{0.2^{0.2}}\right)^5 > 0.9$. If this happens, the first $t$ items will be matched. Therefore, Algorithm {\sc BlindMatch} is a $\frac{2}{0.9} < 2.3$-approximation to $\E[V] \cdot z$, completing the proof.
\end{proof}

\subsection{Approximating {\sc ProFinal}}
\label{sec:final}
\noindent We now approximate {\sc ProFinal} from Lemma~\ref{lem:ProFinal}. 
$\sum_{i = 1}^m \Pr_{h_i \sim H_i}[h_i \geq \ell_{s+1}] \leq 10$ by the definition of the stages. We run Algorithm~\ref{alg:SingleItem}. We randomly sample an item and focus on the item in our algorithm. The probability that item $i$ is sampled is proportional to $\Pr_{h_i \sim H_i}[h_i \geq \ell_{s+1}]$. If item $i$ is sampled, we run an algorithm for the single-item setting (lines $3$ and $4$ in Algorithm~\ref{alg:SingleItem}). The single-item policy is analyzed in Section~\ref{sec:mhrSingle} where it is shown to achieve welfare at least $\frac{1}{2} \cdot \textsc{ProSingle}_i$. 

\begin{algorithm}[htbp]
\DontPrintSemicolon
  For $i = 1, 2, \ldots, m$, set   $q_i \gets \Pr_{h_i \sim H_i}[h_i \geq \ell_{s+1}]$\;
    $i^* \gets $ item $i \in \{1, 2, \ldots, m\}$ with probability $\frac{q_i}{\sum_{i = 1}^m q_i}$\;
    Set the reserve price $p$ so that $\Pr_{v \sim V}[v \geq p] = \frac{1}{\E[H_{i^*}]}$\;
   For each arriving buyer, try selling item $i^*$ with reserve price $p$\;
\caption{\textsc{SingleItem}}\label{alg:SingleItem}
\end{algorithm}

\begin{theorem}
Algorithm~\ref{alg:SingleItem} is a $20$-approximation of $\textsc{ProFinal}$ in expectation.
\label{thm:final_apx}
\end{theorem}
\begin{proof}
By Theorem~\ref{thm:sing_upp} and Theorem~\ref{thm:MHR-Single}, if item $i^* = i$, the algorithm gets $\frac{1}{2} \cdot \textsc{Pro}_i$ in expectation. Thus, the expected welfare achieved by algorithm~\ref{alg:SingleItem} is at least
\begin{align*}
&\sum_{i = 1}^m \frac{\Pr_{h_i \sim H_i}[h_i \geq \ell_{s+1}]}{\sum_j \Pr_{h_j \sim H_j}[h_j \geq \ell_{s+1}]} \cdot \frac{1}{2} \cdot \textsc{ProSingle}_i\\
\geq &\sum_{i = 1}^m \frac{\Pr_{h_i \sim H_i}[h_i \geq \ell_{s+1}]}{10} \cdot \frac{1}{2} \cdot \textsc{ProSingle}_i \geq \frac{1}{20} \cdot \textsc{ProFinal}. \qedhere
\end{align*}
\end{proof}

\subsection{Proof of Theorem~\ref{thm:main}}
\label{sec:mainproof}
Now we are ready to prove our main theorem.
\begin{proof}[Proof of Theorem~\ref{thm:main}]
To summarize our previous discussion:
\begin{enumerate}[label={(\arabic*)}]
\item Theorem~\ref{thm:early_apx} yields a $15.1$-approximation to $\sum_k \textsc{Pro}_k$, where the sum is over odd stages $k \le s$ with $r_k - \ell_k \ge 2$.
\item If we replace ``odd'' with ``even'' in Theorem~\ref{thm:early_apx} and the corresponding algorithm, we have a $15.1$-approximation $\sum_k \textsc{Pro}_k$ over even stages $k$ with $r_k - \ell_k \ge 2$.
\item Theorem~\ref{thm:short_apx} is a $2.3$-approximation to $\sum_k \textsc{Pro}_k$ over stages $k \le s$ with $r_k - \ell_k = 1$.
\item Theorem~\ref{thm:final_apx} yields a $20$-approximation to {\sc ProFinal}.
\end{enumerate}

An algorithm can do one of (1) to (4) with probability $\frac{15.1}{52.5}, \frac{15.1}{52.5}, \frac{2.3}{52.5}$ and $\frac{20}{52.5}$ respectively, yielding a $52.5$-approximation to {\sc Pro}.
\end{proof}

\section{Prophet Inequality for Single Item with MHR Horizon}
\label{sec:mhrSingle}

\noindent In this section, we consider the case where there is $m = 1$ item, and present a proof of Theorem~\ref{thm:single}. The algorithm also serves as our approximation for {\sc ProSingle}, which we use for the overall algorithm with multiple items

We show that the following fixed-price {\em balancing} scheme is a $2$-approximation, and this bound is tight for geometric distributions: 

\begin{quote}
Pretend the item departs uniformly over time at rate $1/\mu$, where $\mu = \E[H]$. Choose a  price $p$ s.t. the rate of acceptance of buyers matches the rate of departure of the item.
\end{quote}

We bound the performance of this policy by using a simple linear programming upper bound on {\sc Pro} that only uses expected values. Though the relaxation is simple, just as in Section~\ref{sec:relax}, it brings out the key insight that the upper bound also behaves like a balancing scheme, except it assumes the item lasts forever when performing the matching. Surprisingly, such a simple relaxation yields the worst-case optimal bound over all MHR distributions.



\begin{theorem} 
\label{thm:sing_upp}
Let $\alpha = 1 - \E_{h \sim H}[(1-\mu^{-1})^{h}]$. Then for $m=1$ items, there is a fixed pricing policy that is $\frac{1}{\alpha}$-competitive. This policy sets the price $p$ such that $\Pr_{X \sim V}[X \geq p] = \frac{1}{\mu}$ where $\mu = \E[H]$.
\end{theorem}
\begin{proof}
First we find an upper bound for \textsc{Pro}. Let $X$ be a random variable with distribution $V$. Consider the following LP:
\begin{equation*}
\begin{array}{ll@{}ll}
\text{maximize}  & \displaystyle\sum_v y(v) \cdot v &\\[3ex]
\text{subject to}& \displaystyle\sum_v y(v) \le 1,\\[3ex]
& \displaystyle y(v) \leq \mu \cdot \Pr_{X \sim V}[X = v], \quad&\forall v.
\end{array}
\end{equation*}
Variable $y(v)$ is the probability that a buyer with realized value $v$ is chosen by prophet. The first constraint requires the item to be sold at most once in expectation. The second constraint says each value can be chosen only when it appears. Both of the constraints are relaxations as they should hold for any realization while the constraints are in expectation. The optimal objective is thus an upper bound for the expected value of the prophet.

Let $\lambda$ be the Lagrange multiplier associated with the first constraint. The partial Lagrangian of the LP is: 
\begin{align*}
\mathcal{L}(\lambda) &= \lambda + \sum_v y(v) \cdot (v - \lambda),\\
y(v) &\leq \mu \cdot \Pr_{X \sim V}[X = v], \quad\forall v.
\end{align*}

The partial Lagrangian is decoupled for each value $v$ and is maximized when $y(v) = \mu \cdot \Pr_{X \sim V}[X = v]$ for any $v \geq \lambda$ and $y(v)=0$ otherwise. For any $\lambda$, this gives us an upper bound on the prophet's welfare. Let $p$ be the value such that $\Pr_{X \sim V}[X \geq p] = \frac{1}{\mu}$. If we set $\lambda = p$, we get the following upper bound for the prophet's value:
\begin{equation*}
\textsc{Pro} \leq \sum_{v \geq p} \mu \cdot v \cdot \Pr[X = v] = \E_{X \sim V}[X \mid X \geq p].
\end{equation*}

Essentially, the prophet pretends that the horizon is infinite and it can always find a buyer with value at least $p$. Now we look at $\textsc{Alg}$ which is an algorithm with a single price $p$. The algorithm has to also consider the event that the horizon ends before the item is matched.
\begin{align*}
\textsc{Alg} &= \E_{X \sim V}[X \mid X \geq p] \cdot  \Pr[\text{a value at least } p \text{ was seen during the time horizon}]\\
&= \E_{X \sim V}[X \mid X \geq p] \cdot \E_{h \sim H} [1 - (1 - \mu^{-1})^h].
\end{align*}
Therefore,
\begin{equation*}
\frac{\textsc{Pro}}{\textsc{Alg}} \leq \E_{h \sim H} [1 - (1 - \mu^{-1})^h]^{-1}. \qedhere
\end{equation*}
\end{proof}


Now, we show that for MHR horizons, this algorithm is $(2-\mu^{-1})$-competitive. The key idea is to use second order stochastic dominance to show that the upper bound is maximized for geometric distributions with the same mean. Somewhat surprisingly, we also show in Theorem~\ref{theo:geoTight}  that this result is tight in the sense that for geometric distributions, no online policy can do better. 

\begin{theorem} 
\label{thm:MHR-Single}
For any MHR distribution with mean $\mu$, $\E_{h \sim H} [1 - (1 - \mu^{-1})^h]^{-1} \leq 2 - \mu^{-1}.$
\end{theorem}

In order to prove the above theorem, we  use \emph{second-order stochastic dominance}.
\begin{definition}
	\rm
	Let $A$ and $B$ be two probability distributions on $\mathbb{R}$. Let $F_A$ be the cumulative distribution function of $A$ and $F_B$ be the CDF of $B$. We say $A$ is \emph{second-order stochastically dominant} over $B$ if for all $x \in \mathbb{R}$,
	\begin{equation*}
	\int_{-\infty}^x (F_B(t) - F_A(t)) \mathrm{d}t \geq 0.
	\end{equation*}
	\label{def:second}
\end{definition}

\begin{prop}
	\label{prop:convex}
	If distribution $A$ is second-order stochastically dominant over $B$, and $A$ and $B$ have the same mean, then for any convex function $f: \mathbb{R} \to \mathbb{R}$, $\E_{x \sim B}[f(x)] \geq \E_{x \sim A}[f(x)]$.
\end{prop}

We now use second order stochastic dominance to show the following.
\begin{lemma}
\label{lem:geo_dom}
Geometric distribution with mean $\mu$ is second-order stochastically dominated by any other MHR horizon distribution with the same mean.
\end{lemma}
\begin{proof}
Let $\phi_c(x) : \mathbb{N}^+ \to \mathbb{R}$ be the following convex function:
	\[ \phi_c(x)  =
	\begin{cases}
	c - x       & \quad \text{if } x \leq c\\
	0  & \quad \text{if } x > c
	\end{cases}
	\]
	where $c$ is a positive integer. Let $G$ be the geometric distribution with mean $\mu$. From Definition~\ref{def:second}, the lemma holds if and only if $\E_{x \sim D}[\phi_c(x)] \leq \E_{x \sim G}[\phi_c(x)]$ for any $c$ and any MHR distribution $D$ with the same mean $\mu$.
	
	We prove this by contradiction. Let $D$ be an MHR distribution with mean $\mu$ which satisfies $\E_{x \sim D}[\phi_c(x)] > \E_{x \sim G}[\phi_c(x)]$ for some $c$. The set of MHR distributions with the same tail after $c$ (the same $\Pr_{x \sim D}[x = x^* \mid x > c]$ for any $x^* > c$) is homeomorphic to a closed and bounded set in $\mathbb{R}^c$, which means it's compact. The function $\E_{x \sim D}[\phi_c(x)]$ is continuous in $D$ under $L^1$-norm, so there is a $D = D^*$ maximizing $\E_{x \sim D}[\phi_c(x)]$ among MHR distributions with the same tail after $c$. This $D^*$ differs from $G$ at some $x \leq c$. Define $q_i = \Pr_{x \sim D^*}[x \geq i + 1 \mid x \geq i]$ and $q = \Pr_{x \sim G}[x \geq i + 1 \mid x \geq i] = 1 - \mu^{-1}$. Because $D^*$ is MHR, $q_i$'s are decreasing. Also $q_1 > q$ as otherwise the mean cannot be $\mu$, and $q_c < q$ as otherwise $\E_{x \sim D^*}[\phi_c(t)] > \E_{x \sim G}[\phi_c(x)]$ cannot hold. Thus there is some $i^* < c$ such that $q_{i^*} > q$ and $q_{i^* + 1} \leq q$.
	
	We are going to show for a pair of small enough $\varepsilon$ and $\varepsilon'$, decreasing $q_{i^*}$ by $\varepsilon$ and increasing $q_{i^* + 1}$ by $\varepsilon'$ such that the mean is preserved will increase $\E_{t \sim D^*}[\phi_c(x)]$. Let $r = 1 + q_{i^* + 2} + q_{i^* + 2}q_{i^* + 3} + q_{i^* + 2}q_{i^* + 3}q_{i^* + 4} + \cdots$. When $\varepsilon \to 0$, we have $\varepsilon(1 + q_{i + 1}r) = \varepsilon' q_i r$. This implies $\varepsilon' q_i - \varepsilon q_{i + 1} > 0$, which means $\E_{x \sim D^*}[\phi_c(x)]$ is increased. It contradicts with the fact that $D^*$ maximizes $\E_{x \sim D^*}[\phi_c(x)]$.
\end{proof}

\begin{proof}{\em (of Theorem~\ref{thm:MHR-Single})}  From Theorem~\ref{thm:sing_upp}, we know $\frac{\textsc{Pro}}{\textsc{Alg}} \leq 1 / \E_{h \sim H}[\phi(h)]$ where $\phi(h) = 1 - (1 - \mu^{-1})^h$ is a concave function. From Lemma~\ref{lem:geo_dom} and Proposition~\ref{prop:convex}, among all MHR distributions $H$ with mean $\mu$, $\E_{h \sim H}[\phi(h)]$ is minimized by a geometric one. For geometric departure with mean $\mu$, $\E_{h \sim H}[\phi(h)] = 2 - \mu^{-1}$.
\end{proof}


\begin{theorem} 
\label{theo:geoTight} 
No online algorithm is better than $(2-\mu^{-1})$-competitive for $m=1$ items when the horizon distribution $H$ is geometric with mean $\mu$.  
\end{theorem}
\begin{proof}
Let $q \in [0, 1)$ be the probability that the process continues after each step. We have $q = 1 - \mu^{-1}$.

Define $\textsc{Alg*}$ as the expected value of the optimal algorithm and $\textsc{Pro}$ as that of the prophet. Let the valuation distribution be: $v_\mathrm{L}$ with probability $1 - p$ and $v_\mathrm{H}$ with probability $p$, $v_\mathrm{L} < v_\mathrm{H}$. At each step, $\textsc{Alg*}$ will set the price to $v_\mathrm{H}$ if it expects to get more than $v_\mathrm{L}$ afterwards. Otherwise it will set the price to $v_\mathrm{L}$. Randomizing over $v_\mathrm{L}$ and $v_\mathrm{H}$ cannot help $\textsc{Alg*}$. Also, because the geometric distribution is memoryless, $\textsc{Alg*}$ will make the same decision every time, \emph{i.e.}, the optimal algorithm is single-threshold. We have
\begin{equation*}
\textsc{Alg*} = \max\left\{v_\mathrm{L} \cdot (1 - p) + v_\mathrm{H} \cdot p, \ v_\mathrm{H} \cdot \frac{p}{1 - q(1-p)}\right\}
\end{equation*}
and
\begin{align*}
\textsc{Pro} = v_\mathrm{H} \cdot \frac{p}{1 - q(1-p)} + v_\mathrm{L} \cdot \left(1 - \frac{p}{1 - q(1-p)}\right).
\end{align*}

When $\mu = 1$ and $q = 0$, the theorem holds because $2 - \mu^{-1} = 1$. Otherwise, we set $v_\mathrm{H}$ so that $\textsc{Alg*}$ is indifferent between its two options. In that case,
\begin{align*}
\lim_{p \to 0} \frac{\textsc{Pro}}{\textsc{Alg*}} &= \lim_{p \to 0} \left(1 + \frac{v_\mathrm{L}}{v_\mathrm{H}} \cdot \frac{1 - q(1 - p)}{p}\right)\\
&= \lim_{p \to 0} \left(1 + \left(\frac{p}{1 - q(1 - p)} - p\right) \cdot \frac{1 - q(1 - p)}{p}\right)\\
&= 1 + q = 2 - \mu^{-1}. \qedhere
\end{align*}
\end{proof}

\section{Lower Bounds for MHR Horizons (Proof of Theorem 1.5)}
\label{sec:multi}
\noindent Next we provide a proof of Theorem~\ref{thm:lb}.
For this, we first show a lower bound of $2$ for any dynamic pricing scheme in the limit when $m$ becomes large, and $1.57$ for any finite $m$. We will subsequently show that no fixed pricing scheme can extract constant fraction of the welfare for $m > 1$ items. For showing these results, we consider a special family of {\em i.i.d.} MHR horizon distributions, which we call {\em low-rate geometric}:  Let $H$ be a geometric distribution with mean $\mu$, so the probability of survival at each step is $q = 1 - \mu^{-1}$. We call $H$ low-rate geometric when $q \to 1^-$. Let $\lambda = 1-q$ be the rate of departure for each item. This goes to $0^+$ when $H$ is low-rate geometric. 

Low-rate geometric distributions correspond to the canonical setting where items are long-lasting, yet their departures are memoryless. In addition to being canonical, the reason we consider this setting is its analytic tractability: It allows us to ignore events where multiple items depart simultaneously, leading to tractable Markov chains for both the prophet and the algorithm. The proof of lower bound of $2$  for large $m$ involves analyzing an interesting time-reversed Markov chain for the prophet's welfare. 

\subsection{Tractable Approximation}
\label{sec:tractable}
\noindent Denote by $\textsc{Alg}^*_m(\lambda)$ the optimal online policy when there are $m$ items and the rate of departures is $\lambda$. Similarly, we define $\textsc{Pro}_m(\lambda)$ to denote the prophet. Since we are considering the limit as $\lambda \rightarrow 0^+$, we will assume throughout that $\lambda < \frac{1}{m}$.

 Define the {\em state} of the system to be $k$ if there are $k$ items in the system. Note that since departures are geometric, any online policy will use a fixed price in each state. The state of the system therefore {\em decreases} over time.  For both of the processes (corresponding to prophet and the optimal algorithm) given the current state is $k$, there is a positive probability that the next state will be $k'$ for any $k' \leq k$.  However, the probability that multiple items depart together (or a match and departures happen together for the algorithm) is extremely small when $\lambda \rightarrow 0^+$. In light of this, we introduce alternative processes for the ease of analysis. 

In an alternative process, we will assume two events (departures, matches) do not simultaneously happen. In other words, for the prophet,  given state $k$, the state transitions to $k-1$ with probability $k \lambda$ per time step. We do not consider state changes due to matching. Instead and equivalently, we will assume that in hindsight, the prophet can optimally match arriving buyers to items that had not departed by that time. Call this prophet $\textsc{Pro}'_m(\lambda)$. For the algorithm, we assume that if the state is $k$, the price is set so that the rate at which a buyer is matched is $\pi_k = \beta_k \lambda k$. Since items also depart at rate $\lambda k$, we will assume the state transitions from $k$ to $k-1$ at rate $(1+\beta_k)\lambda k$. Denote the optimal such algorithm as $\textsc{Alg}'_m(\lambda)$.

\begin{lemma}  (Proved in Appendix~\ref{app:equiv})
\label{lem:equiv}
For any $m \ge 1$: $\frac{\textsc{Pro}_m(\lambda)}{\textsc{Pro}'_m(\lambda)} \to 1$ and $\frac{\textsc{Alg}^*_m(\lambda)}{\textsc{Alg}'_m(\lambda)} \to 1$ as $\lambda \to 0$  
\end{lemma}
Therefore, we will analyze the quantity  $c_m(\lambda) =\frac{\textsc{Pro}'_m(\lambda)}{\textsc{Alg}'_m(\lambda)}$ as the competitive ratio of the algorithm against prophet for any $m, \lambda$ and subsequently take the limit as $\lambda \rightarrow 0^+$. In the remainder of this section, without creating ambiguity we omit the $m$ and $\lambda$ in notation and use $\textsc{Alg}'$ and $\textsc{Pro}'$ instead. 

\subsection{Lower Bound Construction for Dynamic Pricing}
\label{sec:lbgeom}
\noindent To show the lower bounds, we consider the valuation distribution $V$ such that for any $x \in [1,\infty)$, $\Pr_{v \sim V}[v \geq x] = x^{-\alpha}$ where $\alpha \in (1, +\infty)$ is a constant that will be determined later. Note that $\E_{v \sim V}[v]$ is finite. We first give an upper bound for $\textsc{Alg}'$ for this valuation distribution:
\begin{align*}
\textsc{Alg}' \leq \sum_{k = 1}^m \max_{\beta_k} \frac{\beta_k k \lambda}{(1 + \beta_k) k \lambda} \cdot \E[v \mid v \geq F_V^{-1}(1 - \beta_k k \lambda)]
\end{align*}
where $F_V$ is the cumulative distribution function for $V$. The probability of accepting a buyer in state $k$ is at most $\frac{\beta_k k \lambda}{(1 + \beta_k) k \lambda}$ (because acceptance and departure are disjoint events in the alternative process). Simplifying it, we have:
\begin{align*}
\textsc{Alg}'  &\leq \sum_{k = 1}^m \max_{\beta_k} \frac{\beta_k k \lambda}{(1 + \beta_k) k \lambda} \cdot \E_{v \sim V}[v \mid v \geq (\beta_k k \lambda)^{-\frac{1}{\alpha}}]\\
&=  \sum_{k = 1}^m \max_{\beta_k} \frac{\beta_k k \lambda}{(1 + \beta_k) k \lambda} \cdot \frac{\alpha}{\alpha - 1} \cdot (\beta_k k \lambda)^{-\frac{1}{\alpha}}\\
&= \sum_{k = 1}^m (k\lambda)^{-\frac{1}{\alpha}} \cdot \frac{\alpha}{\alpha - 1} \cdot \max_{\beta_k} \frac{\beta_k^{1 - \frac{1}{\alpha}}}{1 + \beta_k}.
\end{align*}
Optimizing over $\beta_k$, we have:
\begin{equation}
\label{eq:alg}
\textsc{Alg}'  = \sum_{k = 1}^m (k\lambda)^{-\frac{1}{\alpha}} \cdot (\alpha - 1)^{-\frac{1}{\alpha}}.
\end{equation}

Now we solve for $\textsc{Pro}'$. Note that for the prophet, we assume the state only changes due to departure of items. Let $p_k(v)$ denote the probability that the item departing in state $k$ is matched to a buyer with valuation at least $v$ by the prophet. In the rest of this section, we call the item departing at state $k$ to be item $k$. We have:

\begin{equation}
\label{eq:lowRatePro}
\textsc{Pro}' = \sum_{k = 1}^m \int_{0}^{+\infty} p_k(v) \mathrm{d} v.
\end{equation}

We now present different bounds for the above quantity depending on whether $m$ is finite, or we are considering the limit $m \rightarrow \infty$.

\subsubsection{Lower Bound for Dynamic Pricing: Finite $m$}
\label{sec:geoAnym}
 This bound is simpler. Clearly, if a buyer with value at least $v$ arrives at state $k$, $\textsc{Pro}'$ always can assign the item $k$ to a buyer with value at least $v$. Therefore,
\begin{align*}
\textsc{Pro}' &\geq \sum_{k = 1}^m   \int_{0}^{+\infty}  \Pr[\text{some buyer with valuation at least } v \text{ arrives in state } k] \mathrm{d} v\\
&= \sum_{k = 1}^m \left(1 + \int_{1}^{+\infty} \frac{(1 - k\lambda)v^{-\alpha}}{v^{-\alpha} + k\lambda -k \lambda v^{-\alpha}} \mathrm{d} v\right).
\end{align*}
Therefore, we have:
\begin{align*}
\liminf_{\lambda \to 0^+} \frac{\textsc{Pro}'}{\textsc{Alg}'} &\geq \lim_{\lambda \to 0^+}  \frac{ \int_1^{+\infty} \frac{1}{1 + k\lambda v^{\alpha}} \mathrm{d} v}{(k\lambda)^{-\frac{1}{\alpha}}(\alpha - 1)^{-1/\alpha}}= \frac{\int_0^{+\infty} \frac{1}{1 + u^\alpha} \mathrm{d} u}{(\alpha - 1)^{-1/\alpha}}  \\
&= \frac{\frac{1}{\alpha}\cdot B(\frac{1}{\alpha}, 1-\frac{1}{\alpha})}{(\alpha - 1)^{-1/\alpha}} = \frac{\frac{\pi}{\alpha} / \sin(\frac{\pi}{\alpha})}{(\alpha - 1)^{-1/\alpha}},
\end{align*}
where $B(\cdot,\cdot)$ is the beta function. Comparing to Equation~(\ref{eq:alg}), we have that $\liminf_{\lambda \to 0^+} \frac{\textsc{Pro}'}{\textsc{Alg}'}$ is maximized at $\alpha = 2$ and in that case, 
$$\frac{\textsc{Pro}'}{\textsc{Alg}'} \ge \pi / 2 \approx 1.5708.$$ 
The bound holds for any $m \in \mathbb{N}^+$.

\subsubsection{Lower Bound for Dynamic Pricing: Large $m$}
\label{sec:geoLargem}
We now consider the more interesting case when $m \to \infty$. We present a tighter lower bound for Equation~(\ref{eq:lowRatePro}). To achieve this goal, we need to to analyze $p_k(v)$  more carefully. Previously, we used the fact that if a buyer with value at least $v$ arrives during state $k$, then a buyer with value at least $v$ will be assigned to the item $k$ by the prophet. However, the prophet might assign a buyer with value at least $v$ to item $k$ even if no such buyer arrives in state $k$. 

It is easy to see that the optimal policy for the prophet is the following: It considers the items in increasing order of realized horizon, and matches each item to the highest valued unmatched buyer arriving no later than the horizon of the item. A buyer with value at least $v$ is matched to the item $k$ if and only if there is an $i \geq 0$ such that between beginning of the state $k + i$ and end of state $k$, at least $i + 1$ buyers with value at least $v$ arrive. Note that in the previous section, we only considered the case of $i = 0$ to give a lower bound for $p_k(v)$. 

\paragraph{Time-reversed Markov Chain.} In order to analyze the new process, we start from the end of state $k$ and go back in time.  There are two possible types of events:
\begin{itemize}
\item An item departs, so that the state increases by $1$ (note we are going back in time); or 
\item A buyer with valuation at least $v$ arrives.
\end{itemize}

We maintain a counter $q$ initially set to $1$. Each time an item departs, we increase $q$ by $1$, and each time a buyer with valuation at least $v$ arrives, we decrease $q$ by $1$. It is easy to see that the item $k$ is matched to a buyer with valuation at least $v$ by the prophet if and only if $q$ reaches $0$, {\em i.e.}, $p_k(v) = \Pr[q = 0 \mbox{ at some time}]$. 

Note that as we are going back in time, when the state is $k+j-1$, the probability an item departs is $(k + j)\lambda$. Similarly, the probability a buyer with valuation at least $v$ arrives is $v^{-\alpha}$. This yields a Markov chain in which when the state is the $(k+j,q)$,  the former event causes the state to become $(k+j+1,q+1)$ and the latter causes the state to become $(k+j,q-1)$. 

As $p_k(v)$'s themselves are hard to analyze, we approximate them by a sequence of functions $\{f_j(x)\}_{j = 0}^\infty$. Each $f_j(x)$ is defined on $[0, 1]$, and it represents the probability that the following random walk ever reaches $0$ in $j$ steps: A point starts at $1$ on the number line. Independently in each step, it goes left by $1$ with probability $x$, and goes right by $1$ otherwise. Note that as $j \rightarrow \infty$, $f_j(x) \rightarrow \min\left(1,\frac{x}{1-x}\right)$. That is, the point-wise limit of $\{f_j(x)\}_{j = 1}^\infty$ as $j \to \infty$ is $f(x) = \min(1, x/(1-x))$. (We slightly abuse notation at $x = 1$ and $f(1) = 1$.)

\begin{lemma}
$p_k(v) \geq f_{j}(v^{-\alpha}/(v^{-\alpha} + (k + j)\lambda))$ for any integer $j \in [1, m - k]$.
\label{lem:dominate}
\end{lemma}
\begin{proof}
From state $k$ to state $k + j$, exactly $j$ departures happen so the process for $p_k(v)$ has at least $j$ moves in this period. For each move, the probability that the counter $q$ decreases is at least $v^{-\alpha}/(v^{-\alpha} + (k + j)\lambda)$. Therefore, we can couple these two processes so that if the random walk ever reaches $0$, the counter must have visited $0$ too.
\end{proof}

We now show that these functions uniformly converge.
\begin{lemma} 
\label{lem:uniform_conv}
$\{\log f_j(x)\}_{j = 1}^\infty$ uniformly converges to $\log f(x)$ on $(0, 1]$. This implies $\forall \varepsilon > 0, \exists k, \forall j > k, \forall x, f_j(x) > (1 - \varepsilon) f(x)$.
\end{lemma}
\begin{proof}
Notice $f_j(x)$ is continuous on $x$ and increasing in $j$. For any $c > 0$, on the compact set $[c, 1]$, each $\log f_j(x)$ is continuous in $x$, and their limit $\log f(x)$ is continuous too. Further, $\log f_j(x)$ is increasing in $j$. By Dini's theorem, the convergence on $[c, 1]$ is uniform.

For any $\varepsilon > 0$, for any $x \in (0, \varepsilon)$ and any $j \geq 1$, $f_j(x) \geq x > \frac{x}{1 - x} \cdot (1 - \varepsilon) = (1 - \varepsilon)f(x)$. Because $\{\log f_j(x)\}_{j = 1}^\infty$ uniformly converges on $[\varepsilon, 1]$, there is a $k$ so that for any $j > k$ and any $x \in [\varepsilon, 1]$, $f_j(x) > (1 - \varepsilon)f(x)$. This completes the proof.
\end{proof}

\medskip
Now we are ready to explicitly compute a lower bound for $\textsc{Pro}'$ as $m \to \infty$. We start with Equation~(\ref{eq:lowRatePro}).
\begin{align*}
\textsc{Pro}' &=  \sum_{k = 1}^m \int_{0}^{+\infty} p_k(v) \mathrm{d} v\\
&\geq \sum_{k = 1}^{m - \sqrt{m}} \int_{0}^{+\infty} p_k(v) \mathrm{d} v\\
&\geq \sum_{k = 1}^{m - \sqrt{m}} \int_{0}^{+\infty} f_{\sqrt{m}}(v^{-\alpha}/(v^{-\alpha} + (k + \sqrt{m})\lambda)) \mathrm{d} v
\end{align*}
where the final inequality follows from Lemma~\ref{lem:dominate}.

Let $c_k = \inf_{x \in (0, 1]} f_k(x) / f(x)$. Then we have:
\begin{align*}
\textsc{Pro}' &\geq c_{\sqrt{m}}\sum_{k = 1}^{m - \sqrt{m}} \int_{0}^{+\infty} f(v^{-\alpha}/(v^{-\alpha} + (k + \sqrt{m})\lambda)) \mathrm{d} v\\
&= c_{\sqrt{m}}\sum_{k = 1}^{m - \sqrt{m}} \int_{0}^{+\infty} \min(1, v^{-\alpha}/((k + \sqrt{m}) \lambda)) \mathrm{d} v\\
&= c_{\sqrt{m}}\sum_{k = \sqrt{m}}^{m} \int_{0}^{+\infty} \min(1, v^{-\alpha}/(k \lambda)) \mathrm{d} v\\
&= c_{\sqrt{m}}\sum_{k = \sqrt{m}}^{m} \left((k\lambda)^{-\frac{1}{\alpha}} + \int_{(k\lambda)^{-\frac{1}{\alpha}}}^{+\infty} v^{-\alpha}/(k \lambda) \mathrm{d} v\right)\\
&= c_{\sqrt{m}}\sum_{k = \sqrt{m}}^{m} \frac{\alpha}{\alpha - 1} \cdot (k\lambda)^{-\frac{1}{\alpha}}.
\end{align*}
When $m \to \infty$, $c_{\sqrt{m}}$ goes to $1$ by Lemma~\ref{lem:uniform_conv}, and $\frac{\sum_{k = \sqrt{m}}^{m} k^{-\frac{1}{\alpha}}}{\sum_{k =1}^{m} k^{-\frac{1}{\alpha}}}$ goes to $1$ too. Thus,
\begin{align*}
\liminf_{m \to \infty} \left ( \liminf_{\lambda \to 0^+} \frac{\textsc{Pro}'}{\sum_{k = 1}^{m} \frac{\alpha}{\alpha - 1} \cdot (k\lambda)^{-\frac{1}{\alpha}}} \right ) &\geq 1.
\end{align*}
Together with the bound for $\textsc{Alg}'$ from Equation~(\ref{eq:alg}), this gives us:
\begin{align*}
\liminf_{m \to \infty} \left (\liminf_{\lambda \to 0^+} \frac{\textsc{Pro}'}{\textsc{Alg}'} \right ) \geq \frac{\frac{\alpha}{\alpha - 1}}{(\alpha - 1)^{-\frac{1}{\alpha}}},
\end{align*}
which reaches its maximum of $2$ at $\alpha = 2$. This completes the proof of Theorem~\ref{thm:lb}.

\subsection{Lower Bound for Fixed Pricing Schemes}
\label{sec:singleLB}
\noindent A natural question is whether there is a single-threshold algorithm that is a constant approximation. Note that this is indeed the case when the horizons $H_i$'s are identical and deterministic; in fact, in this case, the competitive ratio approaches $1$ as $m \rightarrow \infty$. In contrast, when the horizons are not deterministic --- even if they are \emph{i.i.d} geometric, we show that no fixed pricing scheme can be constant-competitive.  This shows the second part of Theorem~\ref{thm:lb}.
\begin{theorem} 
\label{thm:splb}
There exists a family of instances with \emph{i.i.d} geometric horizons, such that any fixed pricing algorithm is $\Omega(\log \log m)$-competitive, where $m$ is the number of items.
\end{theorem}
\begin{proof}
For any $m \geq 2^5$ such that $\log_2 m$ is an integer, consider a geometric horizon distribution $H$ whose mean is $m$: Let $q_m$ be the probability that the horizon is greater than the mean, \emph{i.e.} $q_m = \Pr_{h \sim H}[h > m]$. It is easy to verify $\frac{1}{4} \leq q_m \leq \frac{1}{e}$ since $H$ is geometric. Let the value distribution $V$ satisfy: $\mathrm{supp}(V) = \{1/(q_m^t t^2) \mid t = 3, 4, \ldots, \log_2 m\}$ and $\Pr_{v \sim V}[v \geq 1/(q_m^t t^2)] = q_m^t$ for $t = 3, 4, \ldots, \log_2 m$. Straightforward calculation shows
\begin{align*}
\E_{v \sim V}[v \mid v \geq 1/(q_m^t t^2)] &=  \Theta(1) \cdot \frac{1}{q_m^t} \cdot \sum_{k = t}^{\log_2 m} q_m^k \cdot \frac{1}{q_m^k k^2} \\
& = \Theta(1) \cdot \frac{1}{q_m^t} \cdot \left(\frac{1}{t} - \frac{1}{(\log_2 m) + 1}\right).
\end{align*}
Without loss of generality, for any single-threshold algorithm $\textsc{Sing}$, assume the threshold is $1/(q_m^t t^2)$. We know in time interval $[jm + 1, (j+1)m]$, the expected number of transactions is at most the minimum of the expected number of buyers with valuations at least $1/(q_m^t t^2)$, and the expected number of items alive at the start of the interval. Therefore,
\begin{align*}
\textsc{Sing} &\leq \E_{v \sim V}[v \mid v \geq 1/(q_m^t t^2)] \cdot \sum_{j = 0}^\infty \min(mq_m^t, mq_m^j)\\
&\leq m \cdot \E_{v \sim V}[v \mid v \geq 1/(q_m^t t^2)] \cdot (t + 1) q_m^t/(1 - q_m) = O(m).
\end{align*}

We know from previous discussion that the upper bound from Lemma~\ref{lem:ProEarly} is at most $53 \cdot \textsc{Pro}$. Previously, we set the stages so that about $\frac{1}{2}$ of items depart in each stage. The factor of $\frac{1}{2}$ is not essential and we can change it to any constant strictly between $0$ and $1$, \emph{e.g.} $q_m$. Doing this only costs us a constant.

If we set the reserve price in the interval $[jm + 1, (j + 1)m]$ to be $q_m^j j^2$, we have:
\begin{align*}
\textsc{Pro} &= \Omega(1) \cdot m \cdot \sum_{j = 3}^{(\log_2 m) - 5} q_m^j \cdot \E_{v \sim V}[v \mid v \geq 1/(q_m^j j^2)]\\
&= \Omega(1) \cdot m \cdot \sum_{j = 3}^{(\log_2 m) - 5} 1/j = \Omega(m \log \log m).
\end{align*}

Therefore, $\textsc{Pro} = \Omega(\log \log m) \cdot \textsc{Sing}$ for the constructed family of instances.
\end{proof}

\section{Conclusions}
\label{sec:open}
\noindent In this paper, we consider the setting when items have stochastic horizons. We show a constant-approximation against the prophet when the horizons satisfy the MHR condition. Unlike the classic multi-choice prophet inequalities where the approximation ratio goes to $1$ when the number of items becomes large, we show a $1.57$ (improves to $2$ when the number of items becomes large) approximation lower bound even when the horizons are \emph{i.i.d.} geometric. Our constant is tight for the single-item setting.

We now list several open questions. First, our constant factor for the upper bound ($53$) in the multi-item setting does not match the lower bound ($2$). Closing the gap would be interesting as a future direction.
Next, is it possible to have stochastic horizons in more general prophet-inequality settings such as~\cite{DuettingFKL17}? Finally, it would be interesting to extend our work to the case where items arrive and depart in a stochastic fashion.


\section*{Acknowledgments}
This work is supported by  NSF grants ECCS-1847393, DMS-1839346, CCF-1408784, CCF-1637397, and IIS-1447554; ARL award W911NF-17-1-0094; ONR award N00014-19-1-2268; and research awards from Adobe and Facebook.

\bibliographystyle{acm}
\bibliography{ref}

\appendix
\section{Omitted Proofs}

\subsection{Lower Bound for non-MHR Horizons (Proof of Theorem~\ref{thm:gen_low_any})}
\label{app:gen_low_any}
\noindent We assume $m = 1$ in this proof. The same ideas apply to any $m \ge 1$. Without loss of generality, assume $n = 2^{ck}$ for $c$ that will be fixed later. The horizon is $2^{ci}$ with probability $2^{-i-1}$ for $i = 0, 1, 2, \ldots, k - 1$, and is $n$ with probability $2^{-k}$. Intuitively, there are $k + 1$ possible horizons, where each one is exponentially longer, yet exponentially less probable than the previous one. Denote the valuation distribution by: $a_1$ with probability $p_1$, $a_2$ with probability $p_2$, \ldots, $a_m$ with probability $p_m$ where $a_1 < a_2 < \cdots < a_m$. Here we set $m = ck$, $a_i = 2^{i/c}$ and $p_i = 2^{-i}$ except $p_{ck} = 2^{-ck + 1}$.

Let $\textsc{VPro}$ be any policy that knows realized valuations but not realized horizon, and $\textsc{Pro}$ be the omniscient  prophet. The only information $\textsc{VPro}$ does not know beforehand is the realized horizon, and during execution it cannot do anything once the horizon ends. Therefore it should aim for a specific buyer in advance:
\begin{align*}
\textsc{VPro} &= \E_{v_1, \ldots, v_n} \left[\max_i \left(\sum_{j \geq i}\pi_j\right) M_i(v_1, \ldots, v_n)\right]\\
&\leq 2 \E_{v_1, \ldots, v_n} \left[\max_i \pi_i M_i(v_1, \ldots, v_n)\right],
\end{align*}
where $\pi_i$ is the probability for the horizon to be $2^{ci}$ and $M_i(v_1, \ldots, v_n)$ is the maximum of the first $2^{ci}$ values. Then we have
\begin{align*}
\textsc{VPro} &\leq 4 \cdot \sum_{i = 0}^k 2^i \Pr_{v_1, \ldots, v_n} [\exists j, \pi_j M_j \geq 2^i] \\
 &\leq 4 \sum_{i = 0}^k 2^i \min\left(1, \sum_j \Pr_{v_1, \ldots, v_n}[\pi_j M_j \geq 2^i]\right)\\
&\leq 4 \sum_{i = 0}^k 2^i \min\left(1, \sum_j 2^{cj}\Pr_{v_1, \ldots, v_n}[2^{-j} v_1 \geq 2^i]\right) \\ 
&\leq 4 \sum_{i = 0}^k 2^i \min\left(1, 2\sum_j 2^{cj} 2^{-ci - cj}\right)\\
&\leq 4 \sum_{i = 0}^k 2^i \min\left(1, 2(k+1)2^{-ci}\right) \ \ = O(1)
\end{align*}
when $k = 2^c$. Here the first inequality is an approximation of the Lebesgue integral of $\textsc{VPro}$. The second and third inequalities are union bounds.

On the other hand, we have
\begin{align*}
\textsc{Pro} &\geq \frac{1}{2} \cdot \sum_{i = 0}^k 2^i \sum_{j = 0}^k \pi_j \cdot \Pr_{v_1, \ldots, v_n} [M_j \geq 2^i]\\
&\geq \frac{1}{2} \cdot \sum_{i = 0}^k 2^i \sum_{j = 0}^k 2^{-j} \cdot \min\left((1 - e^{-1}), (1 - e^{-1})2^{cj}\cdot\Pr_{v_1, \ldots, v_n} [v_1 \geq 2^i]\right)\\
&\geq \frac{1}{2} \cdot \sum_{i = 0}^k 2^i \sum_{j = 0}^k 2^{-j} \cdot \min\left((1 - e^{-1}), (1 - e^{-1})2^{cj-ci}\right)\\
&\geq \frac{1}{2} \cdot \sum_{i = 0}^k 2^i 2^{-i} \cdot \min\left((1 - e^{-1}), (1 - e^{-1})2^{ci-ci}\right) \ \ = \Omega(k).
\end{align*}
Here the first inequality is an approximation of the Lebesgue integral of $\textsc{Pro}$. The second inequality uses the fact that: if sum of the probabilities of several independent events is $p \leq 1$, then the union of them happens with probability at least $(1 - e^{-1}) \cdot p$.
As $n = 2^{ck} = 2^{k \log_2 k}$, we know $k = \Theta\left(\frac{\log n}{\log \log n}\right)$.
\subsection{Proof of Lemma~\ref{lem:equiv}}
~\label{app:equiv}
\noindent We only show that  $\frac{\textsc{Pro}_m(\lambda)}{\textsc{Pro}'_m(\lambda)} \to 1$. The proof of the second part that $\frac{\textsc{Alg}^*_m(\lambda)}{\textsc{Alg}'_m(\lambda)} \to 1$  uses a similar argument. We consider the following two processes: the main process based on the actual departure of items in which two departures might happen simultaneously and the alternative process in which at each time step at most one item can depart. The alternative process might modify the number of items in the system at some point during the process with a very small probability. In that case, states of the two processes differ at some point and the two corresponding prophets might achieve different values. Otherwise, they are always at the same state during the process and their values are exactly the same.

There exist two sources of differences (only consider the first time step that they are not at the same state during the process).
The first one which we call type $1$ is as follows: If two departures happen at the same time, alternative process will only consider one of them. In other words, if the main process goes from state $k$ to $k'$ such that $k' < k-1$, the  alternative process will go from $k$ to $k-1$ and will assume there are still $k-1$ items in the system at the next time step. The probability of such a difference for a state $k$ is not more than $\frac{2^k (1-q)^2}{kq^{k-1}(1-q)}$ which goes to $0$ as $q$ approaches $1$. Therefore, using the union bound and the fact that $m$ is finite, the probability of such a difference during the process at some state $k$ denoted by $p_1$ also approaches $0$.

The second source of differences (type $2$) is: If the current state of the main process is $k$ and it remains unchanged after a time step (no departures happens) with a very small probability ($\frac{q^k-1+k(1-q)}{q^k}$), the state of the alternative process will change to state $k-1$ at this time step. The probability of such a difference at state $k$ is $\frac{q^k-1+k(1-q)}{k(1-q)q^{k-1}}$. We can see that this probability goes to $0$ as $q$ approaches $1$ and since $m$ is finite, using the union bound, the probability of such a difference during the process denoted by $p_2$ goes to $0$ as $q$ approaches $1$.
  
Note that the value of $\textsc{Pro}'$ (alternative prophet) can only be greater than $\textsc{Pro}$ (main prophet) if two departures happen at the same time step during the actual departure process (type $1$ difference). However, note that the conditional expectation of $\textsc{Pro}'$ given that such a difference exists is not greater than $\textsc{Pro}'$ (the expected welfare of $\textsc{Pro}'$). Therefore, we have:
\[
\textsc{Pro} \ge (1-p_1) \textsc{Pro}'.
\]
In addition, $\textsc{Pro}$ can be only greater then $\textsc{Pro}'$ if a type $2$ difference exists. Similarly, the conditional expectation of $\textsc{Pro}$ given that a type $2$ difference exists is not greater than $\textsc{Pro}$. Therefore, we also have: 
\[
\textsc{Pro}' \ge (1-p_2) \textsc{Pro}.
\]
Using the last two inequalities,
\[
1-p_1 \le \frac{\textsc{Pro}} {\textsc{Pro}'} \le \frac{1}{1-p_2}.
\]
Using that $p_1$ and $p_2$ both go to $0$, we have $\frac{\textsc{Pro}}{\textsc{Pro}'} \to 1$. 

\end{document}